\documentclass[journal]{IEEEtran}
\ifCLASSINFOpdf
\else
\fi
\usepackage{CJK}
\usepackage{color}
\usepackage{amsmath}
\usepackage{amsfonts,amssymb}
\usepackage{mathrsfs}
\usepackage{amsthm}
\usepackage{algorithm}
\usepackage{algorithmicx}
\usepackage{algpseudocode}
\usepackage{indentfirst}
\usepackage{caption}
\usepackage{ragged2e}
\usepackage{tabularx}

\usepackage{tikz}
\newcommand*{\circled}[1]{\lower.2ex\hbox{\tikz\draw (0pt, 0pt)%
		circle (.4em) node {\makebox[1em][c]{\small #1}};}}

\floatname{algorithm}{Algorithm}

\renewcommand{\raggedright}{\leftskip=0pt \rightskip=0pt plus 0cm}

\usepackage{graphicx} 
\usepackage{epstopdf}
\usepackage[justification=centering]{caption}

\newtheorem{theorem}{Theorem}
\newtheorem{lemma}{Lemma}
\newtheorem{definition}{Definition}

\newtheorem{remark}{Remark}
\newtheorem{corollary}{Corollary}
\newtheorem{proposition}{Proposition}

\begin{document}
%
\title{Leader Selection in Multi-Agent Networks with Switching Topologies via Submodular Optimization}
%
%
%

\author{Kaile Chen, Wangli He,~\IEEEmembership{Senior Member,~IEEE},
        Yang Tang,~\IEEEmembership{Senior Member,~IEEE}, and Wenle Zhang
\thanks{K. Chen, W. He, Y. Tang and W. Zhang are with the Key Laboratory of Advanced Control and Optimization
	for Chemical Processes, Ministry of Education, East China University of Science
	and Technology, Shanghai 200237, China (e-mail: tangtany@gmail.com).}
}
\maketitle

\begin{abstract}
In leader-follower multi-agent networks with switching topologies, 
choosing a subset of agents as leaders is a critical step to achieve desired performances.
In this paper, we concentrate on the problem of selecting a minimum-size set of leaders that ensure the tracking of a reference signal in a high-order linear multi-agent network with a set of given topology-dependent dwell time (TDDT). First, we derive a sufficient condition that guarantees the states of all agents converging to an expected state trajectory. Second, by exploiting submodular optimization method, we formulate the problem of identifying a minimal \emph{leader set} which satisfies the proposed sufficient condition. Third, we present an algorithm with the provable optimality bound to solve the formulated problem. Finally, several numerical examples are provided to verify the effectiveness of the designed selection scheme.
\end{abstract}

\begin{IEEEkeywords}
Leader selection, multi-agent networks, switching topologies, submoular optimization.
\end{IEEEkeywords}

%
\IEEEpeerreviewmaketitle

\section{Introduction}
%
%
%
%
\IEEEPARstart{R}{e}cent decades have witnessed an explosion of research in multi-agent networks (MAN) \cite{Saber-2004}-\cite{Wen-2018}, where the MAN framework is applied to analyze dynamical systems in a rich body of applications, containing the cooperative flight of multiple manned/unmanned combat
aerial vehicles \cite{multiple-leaders} and wireless sensor networks \cite{Yang-2018-tie}. The leader-follower configuration \cite{Song-Qiang-2014} is an important approach in MAN. {\color{black}This technique has emerged due to its contribution to the design of practicable strategies for formation control \cite{Leader-follower}, and is applied to the study of consensus tracking widely \cite{Wen-2014}-\cite{Duan Zhisheng-2018}. Specifically, a subset of agents selected as \emph{leaders} drive a multi-agent network towards the desired objective \cite{Ren-2007-book}. In recent years, several works have unveiled a fact that the suitable placement for leaders would have a major impact on the effectiveness of applying state-of-the-art control techniques \cite{Clark-2014-System-error}-\cite{Patterson-2019}. Therefore, the research interest of exploring systematic approaches to leaders deployment grows gradually. In addition, the study of such mechanisms parallels with controllability research. The significance of leader selection lies in ensuring the desired target with specified leaders as few as possible to impel the feasibility of control.

The early related results are summarized into \cite{2001-review}, where several criteria are presented to show that for control system design input selection plays an important role in realizing expected objectives. Actually, in a leader-follower multi-agent network, leaders act as the role of control inputs. Recent progress on the problem of choosing a minimal set of leaders mainly focuses on optimizing expected performances \cite{Clark-SIAM}. In what follows, three existing frameworks are listed briefly. First, several works study the leader selection problem from a graph-theoretic perspective in \cite{Pequito-robustness} and references therein. 
Especially, in the pioneering work \cite{Liu-nature}, an applicable selection method based on the maximum matching algorithm is presented in the exploration of large-scale complex networks. Second, the authors in \cite{Yang-2016} and \cite{Lin-2014-convex} solve the leader selection problem with a specific objective function, by relaxing the binary constraints into the convex hull. Then, the convex relaxation problem can be solved efficiently via a standard technique, such as the customized interior point method. Third, the submodular optimization technique is an another emerging focus \cite{Tzoumas-2016}-\cite{Summers-2019}. For instance, with such a tool, the leader selection problem for the realization of synchronization within a desired level is studied in \cite{Clark-2017-synchronizarion}. The strength of the submodular optimization scheme lies in the contribution to the establishment of polynomial-time approximation algorithms with the provable optimality bound for computationally prohibitive tasks \cite{Z. Liu}, while the above-mentioned convex relaxation algorithms cannot render such a bound. 
In addition, an overview of submodularity in leader selection is shown in \cite{Clark-magazine}.

In the study of dynamical networks, the situation where the communication topology changes over time is widely considered \cite{survey-2015}-\cite{Wen-2019-switching-2}. In practice, there are quite a few reasons that cause the switching behavior of the interaction topology, such as external disturbances and limitations of sensors \cite{Wen-2015}. The current literature on leader selection, including the above-mentioned works and references therein, mainly considers that the network topology is fixed, while we are interested in the switching case. In \cite{Clark-2014-System-error} where robustness to link noise could be optimized by selecting a \emph{leader set}, the realization of consensus that acts as the requirement of leader selection only depends on connectivity \cite{Ren-2005-tac} in the scenario of switching topologies, due to the first-order agent model. However, except for connectivity, the condition of the dwell time (the time internal between two consecutive switchings \cite{Lin-switching-survey}) is also required to guarantee consensus in the case of the second-order or high-order agent model \cite{Ni-2013-automatica}. Hence, we incorporate the connectivity and the dwell time jointly into the consideration so as to ensure consensus by leader selection. In addition, we relax the condition that each agent is reachable in any predefined graph in \cite{Leader-selection-switching} where the authors investigate leader selection in high-order linear MAS with switching graphs.

In this paper, we study the minimal leader selection problem in the high-order linear multi-agent network with a set of given TDDT. More specifically, we propose a heuristic algorithm to obtain an eligible \emph{leader set} which makes the states of remaining followers converge to the state of leaders, i.e., achieving consensus tracking. 
The main challenge of our work lies in how to reduce the computational complexity, since the solution to a desired \emph{leader set} is a prohibitive task with an exhaustive search if the number of agents is large. In this paper, by leveraging submodularity-ratio \cite{Das-2011}, we present an efficient approximation method with the provable optimality bound to overcome the difficulty.

The main contributions in this paper are listed below.

\begin{itemize}
	\item We derive a sufficient condition to determine the convergence of each follower's state to the states of leaders in the high-order multi-agent network with switching topologies. Furthermore, this condition is based on the case that each system mode corresponding to the predefined topology is considered to be unstable, so that the derived condition could be also applied to the scenario where several or all of the modes are required to be stable. 
	\item We formulate the optimization problem of determining a minimum-cardinality set of leaders that ensure the tracking with a set of given TDDT. Besides, the metric of leader selection is constructed based on the proposed sufficient condition, and is then utilized for evaluating whether an agent could be selected as a leader. 
	\item  By submodularity ratio, we present an efficient algorithm with the provable optimality bound to solve the formulated problem. To reduce the conservativeness, we present another heuristic selection algorithm. 
\end{itemize}

The remaining content of this paper is organized as follows. In Section \uppercase\expandafter{\romannumeral2}, preliminaries are presented. In Section \uppercase\expandafter{\romannumeral3}, the concise problem description is given. In Section \uppercase\expandafter{\romannumeral4}, a heuristic algorithm is shown to determine adequate leaders with a set of given TDDT. In Section \uppercase\expandafter{\romannumeral4}, the effectiveness of the proposed framework is validated by numerical examples. Section \uppercase\expandafter{\romannumeral5} concludes the paper.

\section{PRELIMINARIES}
In this section, we provide notations throughout this paper, and necessary concepts on algebraic graph theory. In addition, two definitions related to submodularity are presented, which are crucial points to describe main results.

\subsection{Notations}
\setlength{\parindent}{1em} $\mathbb{R}^n$ is the $n$-dimensional Euclidean space and $\textbf{1}_n$ is a $n$-dimensional vector with all components being 1. $I_n$ represents the identity matrix of dimension $n$. When a matrix $P$ is positive definite (positive semi-definite), then it is denoted as $ P \succ 0 $ ($ P \succeq  0$). $\text{diag}(A_1\text{,} \, A_2\text{,}\, ...\,\text{,}\,A_n)$ represents a block-diagonal matrix with matrices or scalars $A_i$ on its diagonal, $i=1\text{,}\,2\text{,}\,...\, \text{,}\,n$. $\otimes$ denotes Kronecker product, satisfying $A\otimes B = (A\otimes I_p)(I_n\otimes B)$, where $A\in \mathbb{R}^{m\times n}$, and $B\in \mathbb{R}^{p\times q}$. $|S|$ means the cardinality of a set $S$. $\lambda_i$($A$) represents the $i$th eigenvalue of the matrix $A$. Re($\gamma$) is the real part of the complex number $\gamma$. dist$(x\text{,}\, y)$ represents the Euclidean distance between vectors $x$ and $y$. $\lambda_{\text{max}}(A)$ and  $\lambda_{\text{min}}(A)$ denote the largest and smallest eigenvalues of the matrix $A$ respectively. $\lambda_{\text{r}}(A)$ represents the eigenvalue with the largest real part of the matrix $A$. The superscript $T$ means transpose for real matrices.


\subsection{Algebraic Graph theory}

 We define the digraph structure of a multi-agent network as $\mathcal{G} = \left\{\emph{$\Omega$, E}\right\}$, where the index set of $N$ agents is denoted as $\Omega$ = $\{1\text{,}\, 2\text{,}\, ...\, \text{,}\, N  \}$, while the index set of directed links is shown as $E \in \Omega \times \Omega$. The $N$ agents are divided into \emph{leaders} and \emph{followers}. The former, whose set are denoted by $S$ $\subseteq  \Omega$, act as external control inputs. They are available to the reference signal as well as the state values of their neighbors. The latter are only accessible to the information of adjacent agents. In addition, the neighbors index set of the agent $i$ is defined as  $\mathcal{N}(i) \; \triangleq \left\{j:(j, i)\in\emph{E}\right\}$. $L = [l_{ij}]\in \mathbb{R}^{N \times N}$ means the Laplacian matrix of $\mathcal{G}$, $i, j = 1\text{,}\, 2\text{,}\, ...\, \text{,}\, N$, $i \ne j$, $l_{ij} = -1$ if $(j, i) \in E$ and $l_{ij} = 0$ otherwise, where $(j, i)$ is a directed link from agent $j$ to agent $i$, satisfying
  \begin{align}\notag
	 \begin{split}
	 	\sum\limits_{j=1,j\ne i}^{N} l_{ij}=-l_{ii}.
	 \end{split}
 \end{align}
 Moreover, in this paper, we focus on simple graph only, i.e., without multiple links.

\subsection{Submodularity}

\begin{definition}[\cite{Clark-2014-System-error}]\label{definitionSubmodular}
	\emph{Set $V$ as a finite set. A function $f$: $2^V \rightarrow \mathbb{R}$, is submodular if for any subset of $V$, i.e.,  $S \subseteq T \subseteq V$, and any $v \in V  \backslash T$, such that:}
	\begin{center}
		$ f(S\cup {v}) - f(S) \geq f(T\cup {v}) - f(T)$.
	\end{center}
\end{definition}
This inequality characterizes the submodularity, which is the quantitative measure of the diminishing-return property \cite{Zhang-Edwin. Chong}. It is analogous to concavity of continuous functions, wherein the increment of adding a new component $v \in V \backslash T$ to the set $S$, is larger than or equal to the set $T$. Moreover, a function \emph{f} is submodular, if (-\emph{f}) is supermodular and vice versa.

\begin{definition}[\cite{Das-2011}]\label{submodularity-ratio}
	\emph{$f$: $2^{\Omega} \rightarrow \mathbb{R}$ is a non-negative set function. With respect to a subset $U \subseteq \Omega$ as well as a given constant $k \geq 1$, the submodularity-ratio of $f$ is given by}
	\begin{align}\notag
	\begin{split}
	\gamma_{U,k} =  \mathop{\emph{min}}\limits_{\substack{W \subseteq U \\ W \cap S = \varnothing \\ \left| S \right| \leq k}
	} \;  \frac{\sum_{l \in S}(f(W\cup \{l\}) - f(W))}{f(W\cup S) - f(W)}.
	\end{split}
	\end{align}
\end{definition}
For a general set function $f$, the submodularity-ratio captures its ``distance'' to submodularity. $f$ is submodular if and only if $\gamma_{U,k} \geq 1$, $\forall U,k$, and $f$ is a nondecreasing function, when $\gamma_{U,k} \in [0,1).$  Furthermore, this concept contributes to extending derivation of the provable optimality bound for algorithms, even though $f$ is not exactly submodular. Actually, it is straightforward to see that this definition is applicable to depict the ``distance" of a set function $f$ to supermodularity.

\section{PROBLEM DESCRIPTION}

In this section, we present the dynamics of the multi-agent network and then state the research problem briefly.

\subsection{Dynamics}
The dynamics of individual agent in the high-order linear multi-agent network with switching topologies is described as follows:
\begin{align}\label{1.1}
\begin{split}
{\color{black}\dot{x_i}(t) = Ax_i(t) + B u_i(t) \text{,}}
\end{split}
\end{align}
where
	\begin{align}\notag
		\begin{split}
			u_i(t) = \sum\limits_{j\in\mathcal{N}_{\sigma(t)}(i)} [x_j(t) -  x_i(t)] -d_i K_{\sigma(t)} [x_i(t)-x^*(t)] \text{,}
		\end{split}
	\end{align}
and $x_i(t) \in \mathbb{R}^{n}$ is the state of the agent $i$. $A$ $\in \mathbb{R}^{n\times n}$ is the individual self-dynamics matrix. $B$ is the individual control input matrix and we consider it as $I_N$ for brevity. For the sake of concise statement, we denote $\sigma (t)$ as the switching signal $:[0,\infty) \rightarrow \mathcal{P}$, a right continuous and piece-wise constant mapping. $\sigma (t)$ depicts the time dependence of underlying graphs. $\mathcal{P}$ represents an index set for predefined topologies, i.e., $\mathcal{P} = \left\{\mathcal{G}_1\text{,}\, \mathcal{G}_2\text{,}\, ...\, \text{,}\, \mathcal{G}_m\right\}$, where $m$ is the number of predefined topologies. The time internal between any two consecutive switchings is called as the dwell time $\tau$ \cite{Lin-switching-survey}. For reducing conservation, we consider that the dwell time is not identical but topology-dependent, which is denoted as $\tau_{\mathcal{G}_i}$, $\mathcal{G}_i \in \mathcal{P}$, $i$ = $1$, $2$, $...$, $m$. $\emph{d}_i$ is defined to be $1$ when the agent $i$ is selected as a \emph{leader} and 0 otherwise. The calculation for $K_{\sigma(t)}$ is shown in Section \uppercase\expandafter{\romannumeral4}. Besides, $x^*(t)\in \mathbb{R}^n$ is the state of the given reference signal, where $\dot{x}^*(t)=Ax^*(t)$.

For the convenience of analysis, we define the tracking error state between the agent $i$ and the reference signal as $\epsilon_i(t) = x_i(t) - x^*(t)$. By collecting total tracking error states, we introduce following mathematical expression:
\begin{align}\notag
\begin{split}
\epsilon(t) & = (\epsilon_1^T(t)\text{,}\, \epsilon_2^T(t)\text{,}\,...\,\text{,}\, \epsilon_N^T(t))^T\text{,}  \\
D & = \text{diag}(d_1\text{,}\, d_2\text{,}\, ...\,\text{,}\, d_N) 
\text{.}
\end{split}
\end{align}
Thus, the tracking error dynamics of the multi-agent network is written as the compact form:
\begin{align}\label{system-without-uncertainty}
\begin{split}
\dot{\epsilon}(t) = (I_N\otimes A - L_{\sigma (t)}\otimes I_n -  D \otimes K_{\sigma(t)} )  \epsilon(t) \text{.}
\end{split}
\end{align}
Subsequently, the tracking problem is transformed into the stabilization form, where we consider the \eqref{system-without-uncertainty} as the error system. For brevity, we rewrite the system \eqref{system-without-uncertainty}:
\begin{align}
\begin{split}\label{error-system}
\begin{array}{l}
\dot{\epsilon}(t) =  \mathcal{A}_{\sigma(t)}\epsilon(t) \text{,}
\end{array}
\end{split}
\end{align}
where $\mathcal{A}_{\sigma(t)} = I_N\otimes A - L_{\sigma (t)}\otimes I_n -  D \otimes K_{\sigma(t)} $. Besides, we consider $\mathcal{A}_p$ as a system mode for the $p$th topology, $p \in \mathcal{P}$. In this paper, the interaction topology changes over time, which leads to the mode switching.

\begin{remark}
	\emph{Actually, it is permissible for any system mode to be stable or unstable. If a system mode is considered as the stable case, i.e., $\lambda_{\text{r}}(\mathcal{A}_p) \le0$, it is demanding for the corresponding topology where each agent should be reachable. However, this topology condition is not necessary if there exists no specific constraint, since the topology requirement to ensure the tracking is that each agent is reachable in the union of the directed interaction graphs \cite{Ren-2005-tac}. Therefore, generally,	
		we consider that each mode is unstable, i.e., $\lambda_{\text{r}}(\mathcal{A}_p) > 0$, $\forall p \in \mathcal{P}$, which signifies that it is possible that there exist unreachable agents in any predefined topology. Furthermore, if some or all of the modes are required to be stable, then it would be better to consider multiple \emph{leader sets} so as to reduce the number of unnecessary leaders in each predefined graph. This implies that there are several \emph{leader sets} switching as the communication topology changes over time, and this direction is considered as our future work.  }
\end{remark}

\subsection{Minimal Leader Selection for Tracking}
The problem that we focus on is to select a minimum-size \emph{leader set} $S$ with a set of given TDDT, where the set $S$ determines the configuration matrix $D$ = \text{diag}$(d_1\text{,}\, d_2\text{,}\, ...\,\text{,}\, d_n)$, such that the asymptotic stability of the system \eqref{error-system} can be guaranteed. Then, we further give the description for the problem in terms of optimization form as follows:
	\begin{align}
	\begin{split}
	\bar{\mathcal{P}} \text{1} \;\;\;\;\;   &     \mathop{\text{min}}\limits_{S\subseteq \Omega} \; |S|\\
	\text{s.t.} \; \;  & \text{The system \eqref{error-system} is asymptotically stable} \text{,}
	\\ &	|S| \leq k  \text{,}
	\end{split}
	\end{align}
where $k $ is a given positive integer, as the upper bound for the desired number of leaders. $\bar{\mathcal{P}} \text{1}$ is combinatorial in nature, so that acquiring the solution is a computationally prohibitive task if the number of agents is large. In the next section, we leverage submodularity-ratio to solve $\bar{\mathcal{P}} \text{1}$ efficiently.

\section{THE LEADER SELECTION METHOD}

\setlength{\parindent}{1em}In this section, we mainly propose the method of choosing a minimal set of leaders with a set of given TDDT in order to realize asymptotic stability of the system \eqref{error-system}. In the first subsection, a sufficient condition is derived to guarantee the stability performance, which is equivalent to ensure the convergence of each follower' state to that of leaders. In the next subsection, based on the application of the submodular optimization scheme, we formulate the minimal leader problem. In the remaining subsection, an efficient algorithm is designed with the greedy rule, used for solving the formulated combinatorial optimization problem, and then we prove the optimality bound of the proposed method.

\subsection{The Sufficient Condition for Tracking}

We draw on an existing result, which is regarded as the preparation for our sufficient condition that ensures the system $\eqref{error-system}$ asymptotically stable.

\begin{lemma}[\cite{Xiang-2014}]\label{based-Xiang}
\emph{Given scalars $\eta \geq \eta^* \geq 0$, $\mu \in (0,1)$, $\tau^{\text{max}} \geq \tau^{\text{min}} > 0$, consider the system \eqref{error-system}. If there exists a set of matrices $ P_{p\text{,}i} \succ 0 $, $i = 0$, $1$, $...$, $l$, $p \in \mathcal{P}$, such that $\forall i = 0$, $1$, ..., $l - 1\text{,}\, \forall p\text{,}\,q \in \mathcal{P}$, $  p \ne q \,$,}
	\begin{align}
	 \mathcal{A}_p^TP_{p\text{,}i} + P_{p\text{,}i}\mathcal{A}_p + \psi_{p}^{(i)} - \eta P_{p\text{,}i} \prec 0 \label{Xiang-1}  \text{,} \\
	 \mathcal{A}_p^TP_{p\text{,}i+1} + P_{p\text{,}i+1}\mathcal{A}_p + \psi_{p}^{(i)} - \eta P_{p\text{,}i+1}  \prec 0 \label{Xiang-2}  \text{,} \\
	\mathcal{A}_p^TP_{p\text{,}l} + P_{p\text{,}l}\mathcal{A}_p - \eta P_{p\text{,}l}  \prec 0 \label{Xiang-3}  \text{,} \\
	 P_{q\text{,}0} - \mu P_{p\text{,}l}  \preceq 0  \label{Xiang-4} \text{,} \\
	 \emph{log} \mu + \eta \tau^{\emph{max}}   < 0  \label{Xiang-5} \text{,}
	\end{align}
\emph{where $\psi_p^{(i)} = l(P_{p\text{,}i+1} - P_{p\text{,}i}) / \tau^{\text{min}}$ and $(\mathcal{A}_p - \frac{1}{2}\eta^* I_{Nn}) $ is Hurwitz stable, $\forall p \in \mathcal{P} $, then the system $\eqref{error-system}$ is globally uniformly asymptotically stable (GUAS) under any switching law $\sigma(t)$ $\in$ $\mathcal{D}_{[\tau^{\text{min}}\text{,} \tau^{\text{max}}]}$, where $\mathcal{D}_{[\tau^{\text{min}}\text{,} \tau^{\text{max}}]}$ represents the set of all feasible switching policies with the dwell time $\tau_z \in [\tau^{\text{min}}\text{,} \tau^{\text{max}}] $, $\forall z = 0\text{,}\, 1\text{,}\, 2\text{,}\,...$.}
\end{lemma}

\begin{remark}
\emph{If we regard Lemma \ref{unstable-Xiang} as a sufficient condition directly to ensure the performance of tracking, it is problematic to formulate the minimal leader selection problem, since it is not intuitive to evaluate whether an agent could be selected as a leader by those linear matrix inequalities. In order to solve such a matter, we derive a sufficient condition so as to assure the asymptotic stability of the system \eqref{error-system}, and then we can formulate the leader selection problem via submodular optimization method with a scalar metric. Then, we can design an efficient algorithm to deal with the minimal leader selection problem. Thus, the following result plays a basic role in the construction of the proposed selection scheme.}

\end{remark}

\begin{theorem}\label{unstable-Xiang}
\emph{Given scalars $\eta_p > 0$, $\tau_p^{\text{min}} > 0$ and $\mu_p \in (0\text{,}1)$, consider the system \eqref{error-system} with all unstable modes. If the following conditions hold,
	\begin{align}
		&   \text{Re}(\lambda
		_{\text{r}}(\mathcal{A}_p^{(1)} + \frac{l_p + \varphi}{2\beta \tau_p^{\text{min}}}  I_{Nn})) < 0 \label{Theorem_1} \text{,} \\
	 &	 \text{Re}(\lambda
	 _{\text{r}}(\mathcal{A}_p + \frac{1}{2} (\frac{l_p}{\tau_p^{\text{min}}} - \eta_p) I_{Nn})) <  0 \label{Theorem_2} \text{,} 
	\end{align}
	$\forall p \in \mathcal{P}$, where $\varphi > 0$ is a constant with the sufficient small value,
	\begin{align}
		\beta & = \frac{\lambda_{\text{max}}((\mathcal{A}_p^{(1)})^T + \mathcal{A}_p^{(1)})}{2\text{Re}(\lambda_{\text{r}}(\mathcal{A}_p^{(1)}))} \label{Theorem_3}  \text{,}  \\ \mathcal{A}_p^{(1)} & = \mathcal{A}_p - \frac{1}{2}(\frac{l_p}{\tau_p^{\text{min}}} + \eta_p) I_{Nn}  \text{,} \label{Theorem_4} 
	\end{align}
		then there exist matrices $ I_{Nn} \succ P_{p\text{,}i} \succ 0 $, $i = 0$, $1$, $...$, $l_p$, such that $\forall i = 0$, $1$, $...$, \emph{$l_p - 1$}, $ \forall p \in \mathcal{P}$, satisfying}
	\begin{align}		
		\mathcal{A}_p^T P_{p\text{,}i} + P_{p\text{,}i}\mathcal{A}_p + \phi_{p\text{,}i} - \eta_p P_{p\text{,}i} \prec 0  \text{,}  \label{Theorem_5}  \\
			\mathcal{A}_p^T P_{p\text{,}i+1} + P_{p\text{,}i+1} \mathcal{A}_p + \phi_{p\text{,}i} - \eta_p P_{p\text{,}i+1}\prec  0 \text{,} \label{Theorem_6} 
	\end{align}
\emph{where} $\phi_{p\text{,}i} = l_p(P_{p\text{,}i+1} - P_{p\text{,}i})/\tau_p^{\emph{min}}$. \emph{Furthermore, if there exist matrices $P_{p\text{,}0}$ and $P_{p\text{,}l_p}$, $\forall p \in \mathcal{P}$, $p \ne q$, such that }
	\begin{align}
		P_{q\text{,}0} - \mu_q P_{p\text{,}l_p}  \preceq 0 \text{,} \label{Theorem_7} 
	\end{align}
\emph{then the total tracking error states of the system \eqref{error-system} can converge to zero when the TDDT satisfies $\tau_p \in [\tau_p^{\text{min}}\text{,}\, \tau_p^{\text{max}}]\text{,}\; p \in \mathcal{P}$, where}
	\begin{align}
	\emph{log} \mu_p + \eta_p \tau_p^{\emph{max}}   < 0 \text{.} \label{Theorem_8} 
	\end{align}
\end{theorem}

Prior to showing the proof, a lemma is needed as follows.

\begin{lemma}\label{lemma-eigenvalue}
	\emph{\eqref{Theorem_1} holds, $\forall p \in \mathcal{P}$, and then we have }
	 	\begin{align}\label{lemma-2}
		 \begin{split}
		 {\color{black}\emph{Re}(\lambda_{\emph{r}}(\mathcal{A}_p^{(1)})) < - \frac{l_p + \varphi}{2 \beta \tau_p^{\emph{min}}} \,  \text{,}}
		 \end{split}
		 \end{align}	
\emph{where $\varphi > 0$ is a constant with the sufficient small value.}
\end{lemma}

\begin{proof}
	We refer the readers to Appendix A for the proof in details.
\end{proof}

Here, it is ready to present the proof of our first result.

\begin{proof}[Proof of Theorem \ref{unstable-Xiang}]
Considering Lemma \ref{lemma-eigenvalue} and the fact that $\beta \in (0,1]$ \cite{matrix-lemma-(0-1)}, it is explicit to see
\begin{align}\label{proof-1}
\begin{split}
\frac{(l_p + \varphi)/\tau_p^{\text{min}}}{- 2\beta  \text{Re}(\lambda_{\text{r}}(\mathcal{A}_p^{(1)}))}  < 1 \text{.}
\end{split}
\end{align}	
According to \cite{matrix-lemma-bounded}, since $( (\mathcal{A}_p^{(1)})^T + \mathcal{A}_p^{(1)}) \prec 0$, the solution $P_x \succ 0 $ of the following Lyapunov function
\begin{align}\notag
\begin{split}
(\mathcal{A}_p^{(1)})^T P_x + P_x \mathcal{A}_p^{(1)} + \frac{l_p + \varphi}{\tau_p^{\text{min}}} I_{Nn} = 0 \text{,}
\end{split}
\end{align}	
is bounded by
\begin{align}\notag
\begin{split}
\lambda_{\text{max}}(P_x) \leq \frac{(l_p + \varphi)/\tau_p^{\text{min}}}{- 2\beta  \text{Re}(\lambda_{\text{r}}(\mathcal{A}_p^{(1)}))}  \text{.}
\end{split}
\end{align}	
Thus, due to \eqref{proof-1}, one has
\begin{align}\notag
\begin{split}
\lambda_{\text{max}} (P_x) < 1  \text{.}
\end{split}
\end{align}	
which implies $0 \prec P_x \prec I_{Nn}$. Then, we derive
\begin{align}\notag
\begin{split}
(\mathcal{A}_p^{(1)})^T P_x + P_x \mathcal{A}_p^{(1)} = - \frac{l_p + \varphi}{\tau_p^{\text{min}}} I_{Nn} \prec - \frac{l_p}{\tau_p^{\text{min}}} I_{Nn} \text{,}
\end{split}
\end{align}	
Furthermore, combining with $P_{p \text{,} i} \prec I_{Nn}$, $\forall i = 0$, $1$, ... , $l_p$, then, it is intuitive to derive
		\begin{align}\notag
				\mathcal{A}_p ^T P_{p\text{,}i} + P_{p\text{,}i} \mathcal{A}_p - (\frac{l_p}{\tau_p^{\text{min}}} + \eta_p) P_{p\text{,}i} \prec - \frac{l_p}{\tau_p^{\text{min}}} P_{p \text{,}i+1} \text{.}
		\end{align}
which makes \eqref{Theorem_5} hold. Similarly, in light of \eqref{Theorem_2} as well as Lemma \ref{lemma-eigenvalue}, we have
	\begin{align}\notag
		\begin{split}
	&	{\color{black}\text{Re}(\lambda_{\text{r}}(\mathcal{A}_p +  \frac{1}{2} (\frac{l_p}{\tau_p^{\text{min}}} - \eta_p) ) < - \frac{\varphi}{2\beta_x}  \,  \text{,}}\\ 
		& \beta_x = \frac{\lambda_{\text{max}}((\mathcal{A}_p^{(2)})^T + \mathcal{A}_p^{(2)})}{2\text{Re}(\lambda_{\text{r}}(\mathcal{A}_p^{(2)}))}  \text{,} 
 		\end{split}
	\end{align}
where $\mathcal{A}_p^{(2)} = \mathcal{A}_p +  (l_p/\tau_p^{\text{min}} - \eta_p)/2$. Hence, the solution $I_{Nn} \succ P_x^{'} \succ 0$ of the following Lyapunov function exists, 
	\begin{align}\notag
		(\mathcal{A}_p^{(2)})^T P_x^{'} + P_x^{'}\mathcal{A}_p^{(2)}  + \varphi I_{Nn}= 0 \text{.}
	\end{align}
Then, $\forall i = 0$, $1$, ... , $l_p - 1$, one has 
		\begin{align}\label{P_{i+1}}
			\mathcal{A}_p^T P_{p\text{,}i+1} + P_{p\text{,}i+1} \mathcal{A}_p + \frac{l_p}{\tau_p^{\text{min}}}P_{p\text{,}i+1} - \eta_p P_{p\text{,}i+1} \prec 0	 \text{.}
		\end{align}
Since $P_{p \text{,} i} \succ 0$, $\forall i = 0$, $1$, ... , $l_p$, we see
		\begin{align}\notag
				\mathcal{A}_p^T P_{p\text{,}i+1} + P_{p\text{,}i+1} \mathcal{A}_p + \frac{l_p}{\tau_p^{\text{min}}}P_{p\text{,}i+1} - \eta_p P_{p\text{,}i+1} \prec \frac{l_p}{\tau_p^{\text{min}}} P_{p\text{,}i} \text{.}
		\end{align}
which satisfies \eqref{Theorem_6}. Due to \eqref{P_{i+1}}, it is obvious to see
	\begin{align}\notag	
		\mathcal{A}_p^T P_{p\text{,}i+1} + P_{p\text{,}i+1} \mathcal{A}_p  - \eta_p P_{p\text{,}i+1} \prec -\frac{l_p}{\tau_p^{\text{min}}}P_{p\text{,}i+1} \prec 0	 \text{.}
	\end{align}
With $i = l_p - 1$, we derive
		\begin{align}\notag	
			\mathcal{A}_p^T P_{p\text{,}l_p} + P_{p\text{,}l_p} \mathcal{A}_p  - \eta_p P_{p\text{,}l_p+1} \prec 0	 \text{.}
		\end{align}
Thereby, based on Lemma \ref{based-Xiang}, \eqref{Theorem_5}-\eqref{Theorem_8} make \eqref{Xiang-1}-\eqref{Xiang-5} hold. The proof is complete.
\end{proof}

\begin{remark}
\emph{There are three differences between Lemma \ref{based-Xiang} and Theorem \ref{unstable-Xiang}. First, the latter limits the solution range of the matrices by $ I_{Nn} \succ P_{p\text{,}i} \succ 0 $, $i = 0$, $1$,\,...\,, $l_p$, $p \in \mathcal{P}$, which is helpful to determine the existence of all qualified $P_{p\text{,}i}$. Second, Theorem \ref{unstable-Xiang} utilizes the TDDT, which is less conservative \cite{Lin-switching-survey} compared to the dwell time. In addition, the utilization of TDDT is beneficial to reduce the conservativeness for the leader selection. Third, the proposed condition transforms the existence of \eqref{Xiang-1}-\eqref{Xiang-3} into the determination of \eqref{Theorem_1}-\eqref{Theorem_2}, which is prepared for the construction of a scalar metric in the leader selection algorithm. Furthermore, the exact value of $\eta_p > 0$ is not arbitrary, but it has a specific lower bound with $\eta_p > l_p / \tau_p^{\text{min}}$. Actually, if $\eta_p$ is without such a bound, then it is impossible to make \eqref{Theorem_2} hold when $\lambda_{\text{r}}(\mathcal{A}_p) > 0$. Specially, it is accessible to extend this sufficient condition to considering some or all of stable modes. 
}
\end{remark}

Subsequently, we present a straightforward corollary to deal with the situation, where the system \eqref{error-system} is composed of stable modes and unstable modes. In addition, $\mathcal{S}$ and $\mathcal{U}$ denote the set of stable modes and unstable modes, respectively, where $\mathcal{S} \cup \mathcal{U} = \mathcal{P}$ and $\mathcal{S} \cap \mathcal{U} = \varnothing $.

\begin{corollary}
	\emph{Given scalars $\tau_p^{\text{min}} > 0$, $p \in \mathcal{P}$, $\mu_p \in (1\text{,}+\infty)$, $\eta_p < 0$, $p \in \mathcal{S}$, $\mu_p (0\text{,}1)$, $\eta_p > 0$, $p \in \mathcal{U}$, consider the system \eqref{error-system} with stable modes and unstable modes. If the following inequalities hold}
			\begin{align}\notag	
				& \emph{Re}(\lambda_{\emph{r}}(\mathcal{A}_p)) < \frac{1}{2}\eta_p \text{,} \;\; p \in   \mathcal{S} \text{,} \\
			&	\emph{Re}(\lambda
				_{\emph{r}}(\mathcal{A}_p^{(1)})) < - \frac{l_p + \varphi}{2\beta \tau_p^{\emph{min}}}   \text{,} \;\; p \in \mathcal{U} \text{,}   \notag \\
				&	 \emph{Re}(\lambda
				_{\emph{r}}(\mathcal{A}_p)) <  - \frac{1}{2} (\frac{l_p}{\tau_p^{\emph{min}}} - \eta_p)  \text{,}  \;\; p \in \mathcal{U} \text{,}   \notag
			\end{align}
\emph{then there exist matrices $ I_{Nn} \succ P_{p\text{,}i} \succ 0 $, $i = 0$, $1$, $...$, $l_p$, such that $\forall i = 0$, $1$, $...$, \emph{$l_p - 1$}, satisfying } 
	\begin{align}\label{stable}
			\mathcal{A}_p^T P_{p\text{,}i} + P_{p\text{,}i}\mathcal{A}_p - \eta_p P_{p\text{,}i} \prec 0  \text{,}  \;\; p \in \mathcal{S}\text{,}
	\end{align}	
\emph{and \eqref{Theorem_5}-\eqref{Theorem_6}, $p \in \mathcal{U}$. Thus, if there exist matrices $P_{p\text{,}0}$ and $P_{p\text{,}l_p}$, $\forall p \in \mathcal{P}$, $p \ne q$, such that \eqref{Theorem_7} holds, then the total tracking error states of the system \eqref{error-system} can converge to zero when the TDDT satisfies $\tau_p > -\frac{\text{log}\mu_p}{\eta_p}$, $p \in \mathcal{S}$, $\tau_p \in [\tau_p^{\text{min}}\text{,}\, \tau_p^{\text{max}}]\text{,} \;p \in \mathcal{U}$, where $\tau_p^{\text{max}}$ satisfies \eqref{Theorem_8}.}
	
\end{corollary}

In what follows, we show another corollary, which aims at the system \eqref{error-system} composed of stable modes. Hence, in the next corollary, it is intuitive that $\mathcal{S} = \mathcal{P}$.

\begin{corollary}
	\emph{ Given scalars $\mu_p \in (1\text{,}+\infty)$, $\eta_p < 0$, consider the system \eqref{error-system} with all stable modes. If the condition is fulfilled as follows,  }
		\begin{align}\notag	
			\emph{Re}(\lambda_{\emph{r}}(\mathcal{A}_p)) < \frac{1}{2}\eta_p \text{,} 
		\end{align}
\emph{then there exist matrices $ I_{Nn} \succ P_{p\text{,}i} \succ 0 $, $i = 0$, $1$, $...$, $l_p$, such that $\forall i = 0$, $1$, $...$, \emph{$l_p - 1$}, satisfying \eqref{stable}, $\forall p \in \mathcal{P}$. Hence, if there exist matrices $P_{p\text{,}0}$ and $P_{p\text{,}l_p}$, $\forall p \in \mathcal{P}$, $p \ne q$, such that \eqref{Theorem_7} holds, then the total tracking error states of the system \eqref{error-system} can converge to zero when the TDDT satisfies $\tau_p > -\frac{\text{log}\mu_p}{\eta_p}$.}
	
\end{corollary}

The proof of two corollaries can be obtained based on Theorem \ref{unstable-Xiang} as well as \cite{Xiang-2014}, and then they are omitted here. Thereby, the proposed sufficient condition characterizes the situation that even if each mode of the system  $\eqref{error-system}$ is unstable, it is still possible to realize tracking under an appropriate $\sigma(t)$. It is required to point out that the feasibility of \eqref{Theorem_1}-\eqref{Theorem_2} poses the possibility of satisfaction for \eqref{Theorem_7}-\eqref{Theorem_8}, so that the system $\eqref{error-system}$ is GUAS with the TDDT  $\tau_p \in [\tau_p^{\text{min}}\text{,}\, \tau_p^{\text{max}}]\text{,}\; p \in \mathcal{P}$. As a result, we consider \eqref{Theorem_1}-\eqref{Theorem_2} as decisive factors. Actually, they serve as preconditions to the existence of an eligible \emph{leader set} in the proposed algorithm. 

For the convenience of the leader selection metric construction, we present following proposition to combine \eqref{Theorem_1}-\eqref{Theorem_2} into one constraint.

\begin{proposition}\label{proposition-2}
	\emph{Due to $\beta \in (0,1]$ \cite{matrix-lemma-(0-1)}, we set $\beta = 1$. Then, considering \eqref{Theorem_1}, we obtain
		\begin{align}\label{proposition-2-1}
		\text{Re}(\lambda_{\text{r}}(\mathcal{A}_p^{(1)}))  < - \frac{l_p + \varphi}{2  \tau_p^{\text{min}}} \text{.}
		\end{align}
Thus, if \eqref{Theorem_2} holds, then \eqref{proposition-2-1} is satisfied.
	}
\end{proposition}

\begin{proof}
	Based on Lemma \ref{lemma-eigenvalue}, for \eqref{Theorem_2}, it is obvious to see
	\begin{align}\notag
	\text{Re}(\lambda_{\text{r}}(\mathcal{A}_p)) < \frac{1}{2} (\eta_p - \frac{l_p}{\tau_{\text{min}}}) \text{.}
	\end{align}
	Similarly, for \eqref{proposition-2-1}, we have
	\begin{align}\notag
	\text{Re}(\lambda_{\text{r}}(\mathcal{A}_p)) < \frac{1}{2} (\eta_p + \frac{l_p}{\tau_{\text{min}}})  - \frac{l_p + \varphi}{2 \tau_p^{\text{min}}} = \frac{1}{2}(\eta_p - \frac{\varphi}{\tau_p^{\text{min}}}) \text{.}
	\end{align}
	The $\varphi > 0 $ is a constant with the sufficient small value, and then Proposition \ref{proposition-2} holds intuitively. The proof is complete.
\end{proof}

\begin{remark}
	\emph{Naturally, whatever the value of $\beta$ is, we still can obtain one condition that satisfies \eqref{Theorem_1}-\eqref{Theorem_2} simultaneously based on Proposition \ref{proposition-2}. It is obvious to see that \eqref{Theorem_1} is almost equal to \eqref{Theorem_2} when $\beta = 0.5$. In fact, based on the proposed algorithm in the last subsection, after acquiring the configuration matrix $D$ as well as $K_p$, $\forall p \in \mathcal{P}$, the value of $\beta$ could be obtained by calculations, and $\beta$ should be verified. If the computed value of $\beta$ is less than $0.5$, then we reduce the setting value such as $\beta = 0.4$ to operate the selection method all over again until the calculated value of $\beta$ is more than or equal to the setting value, and without any operation when the computed value of $\beta$ is more than $0.5$. }
\end{remark}


\subsection{A Metric for Leader Selection}
In this subsection, we establish the metric for the leader selection method with the form of the $\gamma$-submodular function. Besides, the $\gamma$-submodular means the function with respect to $\gamma$ submodularity-ratio. Then, we finish the metric construction by introducing the lemma below:
\begin{lemma}[\cite{Z. Liu}]\label{lemma-K}
\emph{For a linear system such as}
	\begin{equation}
		\left\{
		\begin{aligned}\label{closed-loop}
		\dot{\hat{x}} (t)& = \hat{A} \hat{x}(t) + \hat{B} \hat{u}(t) \text{,} \\
		\hat{y}(t) & = \hat{C} \hat{x}(t) \text{,}
		\end{aligned}
		\right.
	\end{equation}
\emph{it is determined as a fully observable control system, where $\hat{x}(t) \in \mathbb{R}^n$ is the system state. There exists a feedback control matrix $\hat{K}$, if all eigenvectors $v_i$ of $\hat{A}$ with eigenvalues $\lambda_i$ satisfying Re$(\lambda_i) \geq \hat{\lambda}$, lie in the span of the controllability matrix $\mathcal{C}(\hat{A}, \hat{B})$. It ensures $\text{Re}(\lambda_i(\hat{A} - \hat{B} \hat{K})) < \hat{\lambda}$ for the closed-loop system \eqref{closed-loop}, where $\hat{\lambda}$ is a given constant. Moreover, span($\mathcal{C}(\hat{A}, \hat{B})$) = span$(W(\hat{A}, \hat{B}))$ $\cite{span(W(S))=span(C(S))}$, where}
	\begin{align}\notag
		\begin{split}
		  \mathcal{C}(\hat{A}, \hat{B}) & = [ \hat{B} \;\; \hat{A} \hat{B}\;\; ... \;\; \hat{A}^{n-1}\hat{B}] \, \text{,}\\
		 W(\hat{A}, \hat{B}) & = \int_{t_0}^{t_1}e^{\hat{A}(t-t_0)}\hat{B} \hat{B} ^Te^{\hat{A}^T(t-t_0)}dt \text{,}
		\end{split}
	\end{align}
\emph{\text{for some} $t_1  > t_0$.	}
\end{lemma}
\setlength{\parindent}{1em}Thus, due to Lemma \ref{lemma-K}, {\color{black} we rewrite the parameter matrices in \eqref{Theorem_2}}, $p \in \mathcal{P}$:
\begin{align}\notag
\begin{split}
  \hat{\mathcal{A}}_p = & \;I_N \otimes A - L_p\otimes I_n  + \frac{1}{2} (\frac{l_p}{\tau_p^{\text{min}}} - \eta_p ) I_{Nn}  \text{,}\\
  \hat{B} = & \; D\otimes I_n \text{,} \;
 \hat{K}_p =  I_N \otimes K_p \text{.}
\end{split}
\end{align}
Thus, the metric is constructed as
\begin{align}\notag
\begin{split}
 f \triangleq \sum_{p\in \mathcal{P}}  \sum_{i:\text{Re}_{f}}  \text{dist}^2(v_i, \text{span}(W(\hat{\mathcal{A}}_p, \hat{B}))) \text{,}
\end{split}
\end{align}
where $\text{Re}_{f} = $ $ \text{Re}(\lambda_i(\hat{\mathcal{A}}_p))\geq 0$, $p \in \mathcal{P}$. Re$_{f}$ is an eigenvalue of $\hat{\mathcal{A}}_p$ with the condition of $\text{Re}(\lambda_i(\hat{\mathcal{A}}_p))\geq 0$, and $v_i$ is the corresponding eigenvector.

\begin{remark}
	\emph{{\color{black} In light of Lemma \ref{lemma-K} and Proposition \ref{proposition-2}, it is inferred that if $f = 0$ then there exists a set of $\hat{K}_p$, which makes \eqref{Theorem_1}-\eqref{Theorem_2} hold. Thus, if $\eqref{Theorem_7}$-\eqref{Theorem_8}} are also satisfied, then the tracking of a reference signal can be guaranteed if $\tau_p \in [\tau_p^{\text{min}}\text{,}\, \tau_p^{\text{max}}]$, $\forall p \in \mathcal{P}$. The strength of such transformation is that we can assure the existence of $K_p$ by \eqref{Theorem_2} after leader selection, instead of designing the exact values of $K_p$ before the selection algorithm, $p \in \mathcal{P}$. 
	}
\end{remark}

Thus, we rewrite $\bar{\mathcal{P}}$1 as follows.
\begin{align}
\begin{split}
\check{\mathcal{P}} \text{1} \;\;\;\;\;       & \mathop{\text{min}}\limits_{S\subseteq \Omega} \; |S|\\
\text{s.t.} \; \; &  f = 0 \, \text{,} \\
\begin{split}
&  \eqref{Theorem_7}-\eqref{Theorem_8}  \, \text{,}
\end{split}                   \\
& 	|S| \leq k  \text{.}
\end{split}
\end{align}

\setlength{\parindent}{1em}Here, the metric $f$ construction is completed. Prior to showing our further result, some notations are listed below:
\begin{align}\notag
\begin{split}
 C_{\Omega, p}(\hat{\mathcal{A}}_p, B_{\Omega})& = [B_{\Omega}\text{,}\; \hat{\mathcal{A}}_p B_{\Omega}\text{,} \;...\text{,}\; \hat{\mathcal{A}}_p^{Nn - 1}B_{\Omega}]  \text{,}\\
 C_p(\hat{\mathcal{A}}_p\text{,}\, \hat{B}_S) & = [\hat{B}_S \text{,}\; \hat{\mathcal{A}}_p\hat{B}_S \text{,} \;...\text{,}\; \hat{\mathcal{A}}_p^{Nn - 1}\hat{B}_S] \text{,}\\
  B_{\Omega}  = I_{Nn}\text{,} \; \hat{B}_S &= D \otimes I_n  \text{,} \; p \in \mathcal{P} \text{,}  \\
 \bar{C}_p & = (C_{\Omega, p} P_t)^T (C_{\Omega, p}P_t)/(Nn)  \text{,} \\
 \lambda_{\text{min}}(\bar{C}_p\text{,}\,k+|U|) & =\mathop{\text{min}}\limits_{ S:|S| = k + |U| }\lambda_{\text{min}}(\bar{C}_{S,p}) \text{,}
\end{split}
\end{align}
where $D$ is the diagonal matrix determined by \emph{leader set} $S$, where the $i$th diagonal element of $D$ is 1 if $i$th agent is selected as \emph{leader} and 0 otherwise. $P_t\in \mathbb{R}^{(Nn)^2\times(Nn)^2} $ is a nonsingular matrix, leading to each column of $(C_{\Omega\text{,}p}P)$ have norm 1. $\bar{C}_{S\text{,}\,p}$ is derived from $\bar{C}_p$ by removing all zeroes rows and columns. A vector $v$ with $\left\| v \right\|_2 = 1$, by referring to $\cite{Z. Liu}$, we have
\begin{align}\notag
\begin{split}
f_{v\text{,}p} = \text{dist}^2(v\text{,}\,\text{span}(C_p(\hat{\mathcal{A}}_p\text{,}\, \hat{B}_S)) = 1 - g_{v,p} \; \text{,} \;\; p \in \mathcal{P} \text{,}
\end{split}
\end{align}
where
\begin{align}\notag
\begin{split}
g_{v,p}  & = Nn \tilde{v}^T \bar{C}_{S,p}^{-1}\tilde{v} \text{,}  \\
 \; \tilde{v} & =  \bar{C}_{S,p}^{'}v/Nn  \text{,} \; \\ 
 \bar{C}_{S,p} & =  \bar{C}_{S,p}^{'T} \bar{C}_{S,p}^{'}/Nn.
\end{split}
\end{align}
Then, due to \cite{Das-2011}, the submodularity-ratio $\gamma_{U\text{,}k\text{,}p}'$ of $g_{v\text{,}p}$ is bounded by
\begin{align}\notag
\begin{split}
\gamma_{U\text{,}k\text{,}p}' \geq \lambda_{\text{min}}(\bar{C}_p\text{,}\, k + |U|) \geq \lambda_{\text{min}}(\bar{C}_p) \, \text{.}
\end{split}
\end{align}

Here, we are ready to show our further result, which is helpful to acquire the provable optimality bound of the proposed algorithm. Concisely, the submodularity-ratio of $f$ is bounded by $\gamma_{U\text{,}k}'$.

\begin{theorem}\label{submodularity-ratio-f1}
\emph{The submodularity-ratio $\gamma_{U\text{,}k}$ of $f$ is bounded by}
	\begin{align}\notag
	\begin{split}
	\gamma_{U\text{,}k}   \geq  （\mathop{\emph{min}}\limits_{p}）_{p \in \mathcal{P}} \lambda_{\emph{min}}(\bar{C}_p\text{,}\, k + |U|)
	\geq （\mathop{\emph{min}}\limits_{p}）_{p \in \mathcal{P}} \lambda_{\emph{min}}(\bar{C}_p)
	\end{split} \, .
	\end{align}
\end{theorem}

\begin{proof}
	By definition,
	\begin{align}\notag
	\begin{split}
	\gamma_{U\text{,}k\text{,}p}'  = \mathop{\text{min}}\limits_{ \substack{W \subseteq U \\ W \cap S = \varnothing \\ |S| \leq k} } \frac{\sum_{l\in S}(g_{v\text{,}p}(W\cup \{l\}) - g_{v\text{,}p}(W))}{g_{v\text{,}p}(W\cup S  ) - g_{v\text{,}p}(W)}\\
	= \mathop{\text{min}}\limits_{ \substack{W \subseteq U \\ W \cap S = \varnothing \\ |S| \leq k }} \frac{\sum_{l\in S}(f_{v\text{,}p}(W\cup \{l\}) - f_{v\text{,}p}(W))}{f_{v\text{,}p}(W\cup S  ) - f_{v\text{,}p}(W)} .
	\end{split}
	\end{align}
	Then, we derive the submodularity-ratio $\gamma_{U\text{,}k}$ of $f$ bounded by $\gamma_{U\text{,}k}'$. It is explicit that
	\begin{align}\notag
	\begin{split}
	f = \sum_{p \in \mathcal{P}}\sum_{i:\text{Re}(\lambda_i(\hat{A}_p))\geq 0} f_{v_i,p} \, .
	\end{split}
	\end{align}
	Then, we obtain
	\begin{align}\notag
		\gamma_{U\text{,}k} & =  \mathop{\text{min}}\limits_{ \scriptstyle  W \subseteq U \atop {\scriptstyle  W \cap S = \varnothing \atop {\scriptstyle |S| \leq k \atop}}  } \frac{    \sum_{l \in S}(f(W \cup \{ l\})) - f(W)}{f(W \cup S) - f(W)}\\
		& \geq  \mathop{\text{min}}\limits_{ \scriptstyle  W \subseteq U \atop {\scriptstyle  W \cap S = \varnothing \atop {\scriptstyle |S| \leq k \atop}} } \mathop{\text{min}}\limits_{p} \frac{\sum_{i:\text{Re}_{f}} \sum_{l \in S} f_{v_i\text{,} p}^{\bigtriangledown}   }{\sum_{i:\text{Re}_{f}} \sum_{l \in S} 	f_{v_i\text{,} p\text{,}S}^{\bigtriangledown}}  \notag \\
		& \geq   \mathop{\text{min}}\limits_{ \scriptstyle  W \subseteq U \atop {\scriptstyle  W \cap S = \varnothing \atop {\scriptstyle |S| \leq k \atop}} } \mathop{\text{min}}\limits_{p} \mathop{\text{min}}\limits_{i:\text{Re}_{f}}  \frac{ \sum_{l \in S} f_{v_i\text{,} p}^{\bigtriangledown} }{ \sum_{l \in S} 	f_{v_i\text{,} p\text{,}S}^{\bigtriangledown}}\notag \\
		& \geq  \mathop{\text{min}}\limits_{p} \mathop{\text{min}}\limits_{ \scriptstyle  W \subseteq U \atop {\scriptstyle  W \cap S = \varnothing \atop {\scriptstyle |S| \leq k \atop}} } \frac{\sum_{l\in S}(f_{v\text{,}p}(W\cup \{l\}) - f_{v\text{,}p}(W))}{f_{v\text{,}p}(W\cup S  ) - f_{v\text{,}p}(W)} \notag \\
		& \geq  \mathop{\text{min}}\limits_{p}\lambda_{\text{min}}(\bar{C}_p\text{,}\, k + |U|) \notag  \\
		& \geq \mathop{\text{min}}\limits_{p} \lambda_{\text{min}}(\bar{C}_p) \text{,} \notag
	\end{align}
	where
	\begin{align}\notag
	\begin{split}
	\text{Re}_{f} & = \lambda_i(\hat{A}_p))\geq 0\text{,} \;  p\in \mathcal{P} \, \text{,} \\
	f_{v_i\text{,} p}^{\bigtriangledown} & = f_{v_i\text{,}p}(W\cup \{ l\}) - f_{v_i\text{,}p}(W) \, \text{,} \\
	f_{v_i\text{,}p\text{,}S}^{\bigtriangledown} & = f_{v_i\text{,}p}(W \cup \{ S\}) - f_{v_i\text{,}p}(W) \, \text{.}
	\end{split}
	\end{align}
	The proof is complete.
\end{proof}

 Here, we finish total preparation for our algorithm, which is utilized for figuring out the solution to $\check{\mathcal{P}} 1$ with the provable optimality bound.


\subsection{The Leader Selection Algorithm}

In this subsection, we propose a heuristic algorithm with the greedy rule to select a minimum-size \emph{leader set} $S$, which leads to $f = 0$ as well as ensures \eqref{Theorem_7}-\eqref{Theorem_8} in Theorem \ref{unstable-Xiang}. Thereby, the tracking of a reference signal can be guaranteed with a set of given TDDT. Specially, in order to ensure the tracking, it is a necessary condition that each agent is reachable in the union of the directed interaction graphs \cite{Ren-2005-tac}. Therefore, we take $S_{0}$ as the index set of agents, which are unreachable in the union of the directed interaction graphs. Then, we consider $S_{0}$ as the initial \emph{leader set} in the algorithm, $T_{\text{min-max}} = \{ \tau_p^{\text{min}}\text{,}\tau_p^{\text{max}}\}$ as the set of given TDDT, and $Q = \{l_p$, $\mu_p$, $\eta
_p\}$ as the set of parameters in Theorem \ref{unstable-Xiang}, $\forall p \in \mathcal{P}$.

\begin{algorithm}[H]
	\caption{Algorithm for selection of a minimum-size \emph{leader set} $S$ with a set of given TDDT to assure the tracking of a reference signal}
	\begin{algorithmic}[1]\label{algorithm}
		\Require  The agents index set $ \Omega$, the metric $f$, a constant $k$, $T_{\text{min-max}}$ and $Q$
		\Ensure The \emph{leader set} $S$
		\Procedure{MinSet $(S, f)$}{}
		\State \textbf{Initialization:} \;\;\;\; $S  \leftarrow S_{0}  $ ,  $z  \leftarrow  0  $
		\quad  \quad  \While{$f > 0$ }
		\quad  \quad       \For{ $v_x$ $\in \Omega \backslash  S$}
		\quad  \quad      \State $F_{v_x}$  $\leftarrow$  $f(S) - f(S\cup \{v_x\})    $ 
		\quad  \quad    \EndFor\\
		\quad  \quad \quad \quad $v^*$   $\leftarrow$  arg $\text{max}_{v_x}$  $F_{v_x}$\\
		\quad  \quad\quad \quad  $S$   $\leftarrow$   $S \cup \{v^*\}$
		\quad  \quad    \EndWhile\\
		\quad  \quad  \textbf{if}  $|S| \leq k$ and \eqref{Theorem_7}-\eqref{Theorem_8} hold with a set of $K_p$ \\
		\quad \quad \quad \textbf{return} $S$\\
		\quad  \quad  \textbf{else} \\
		\quad  \quad  \quad $z = z + 1$ \\
		\quad  \quad  \quad  switch to next step \\
		\quad  \quad  \textbf{if} $f = 0$ with a new set $Q$  \\
		\quad  \quad  \quad switch to step 10  \\
		\quad  \quad   \textbf{else}  \\
		\quad  \quad   \quad switch to next step  \\
		\quad  \quad   \textbf{if} $z$ reaches a specified maximum number $z_{\text{max}}$ \\
		\quad  \quad   \quad  \textbf{return} ``None with such a $T_{\text{max-min}}$"  \\
		\quad  \quad   \textbf{else}  \\
		\quad  \quad   \quad switch to step 3 with $S = S_{0}$
		\EndProcedure
	\end{algorithmic}
\end{algorithm}

\begin{remark}
	\emph{From the pseudo-code of Algorithm 1, a candidate \emph{leader set} $S$ is obtained after step 9, and then $K_p$ can be acquired via \eqref{Theorem_2} with the configuration matrix $D$. It is not accessible to acquire a feasible \emph{leader set} with the arbitrarily given $T_{\text{max-min}}$. In accordance with \cite{Xiang-2014}, by switching behaviors, the ability of the tracking error state compensation is limited, and then the solution returned by the algorithm may be none. Besides, the lower bound $k_{\text{min}}$ of the integer $k$ depends on the number of unreachable agents in the union of the directed interaction graphs before leader selection and $k_{\text{min}} = 1$ otherwise. There is a predefined condition that is required to be satisfied:} $\; \mathop{\emph{max}}\limits_{ p \in \mathcal{P} } \lambda_{\emph{r}}(\mathcal{\hat{A}}_p) \ge 0 \text{,}$
\setlength{\parindent}{0em}\emph{which captures the rationality of the selection scheme. If this condition does not hold, then it causes $f = 0$ when $D = \varnothing$, and it is impossible to ensure the tracking obviously. 
	}
\end{remark}

On the consideration of Theorem \ref{submodularity-ratio-f1}, the provable optimality bound of Algorithm 1 is given as follows, which is served as measuring the optimality of resulting \emph{leader set}.

\begin{proposition}\label{proposition-1}
	\emph{Let the optimal solution of $\check{\mathcal{P}} 1$ be represented by $S^*$, where $|S^*| \geq k_{\text{min}}$. Consider $S$ = $\{s_1\text{,}\,s_2\text{,}\,...\,\text{,}\,s_{|S|}\}$ as the solution returned by the Algorithm 1 in the first $|S|$ iterations. Then, we have}
	\begin{align}\notag
	\begin{split}
	e^{- \frac{k_{\emph{min}}}{k} \gamma_{\Delta}} f(\varnothing) \geq  f(S^{t-1})\, \text{,}
	\end{split}
	\end{align}
	\emph{where $\gamma_{\Delta}$ = $（\mathop{\text{min}}\limits_{p}）_{p \in \mathcal{P}} \lambda_{\text{min}}(\bar{C}_p\text{,}\,2|S|)$ and $S^{t - 1}$ denotes the result of Algorithm 1 at the second-to-last iteration.}
\end{proposition}

\begin{proof}
	By referring to \cite{Z. Liu} as well as the definition of the submodularity-ratio, we obtain
	\begin{align}\notag
	\begin{split}
	\frac{\sum_{l\in S}f(\varnothing) - f(\{l\})}{f(\varnothing) - f(S)} \geq \gamma_{S\text{,}|S|} \geq  \gamma_{\Delta} .
	\end{split}
	\end{align}
	Combining with the greedy rule of Algorithm 1, we have
	\begin{align}\notag
	\begin{split}
	|S|(f(\varnothing) - f(s_1)) \geq \gamma_{\Delta} (f(\varnothing) - f(S)) \, \text{,}
	\end{split}
	\end{align}
	meaning that
	\begin{align}
	\begin{split}
	(1 - \frac{\gamma_{\Delta}}{|S|})(f(\varnothing) - f(S)) \geq f(s_1) - f(S)  \, \text{.} \label{optimality-1}
	\end{split}
	\end{align}
	Subsequently, we derive that
	\begin{align}\label{optimality-2}
	\begin{split}
	|S| (f(s_1) - f({s_1\text{,}\, s_2})) & \geq \gamma_{\Delta}(f(s_1) - f(S\cup \{s_1\})) \\
	& \geq \gamma_{\Delta}( f(s_1) - f(S)) \, \text{.}
	\end{split}
	\end{align}
	Furthermore, $\eqref{optimality-2}$ is equivalent to
	\begin{align}\notag
	\begin{split}
	(1 - \frac{\gamma_{\Delta}}{|S|}) ( f(s_1) - f(S)) \geq f({s_1\text{,}\, s_2})  - f(S) .
	\end{split}
	\end{align}
	Considering the $\eqref{optimality-1}$ jointly, we obtain
	\begin{align}\notag
	\begin{split}
	(1 - \frac{\gamma_{\Delta}}{|S|})^2(f(\varnothing) - f(S)) \geq  f({s_1\text{,}\,s_2}) - f(S) .
	\end{split}
	\end{align}
	Thus, by induction method, it is explicit that
	\begin{align}\notag
	\begin{split}
	(1 - \frac{\gamma_{\Delta}}{|S|})^{|S^*|}(f(\varnothing) - f(S)) \geq f(s_1\text{,}\, s_2\text{,}\,...\,\text{,}\,s_{|S^*|}) - f(S)  \, \text{,}
	\end{split}
	\end{align}
	signifying further,
	\begin{align}\notag
	\begin{split}
	(1 - \frac{\gamma_{\Delta}}{|S|})^{|S^*|}(f(\varnothing) - f(S)) \geq  f(S^{t - 1}) - f(S) \, \text{.}
	\end{split}
	\end{align}
	It is noted that $f(S) = 0$, and then we obtain
	\begin{align}\notag
	\begin{split}
	|S^*|\text{log}(1 - \frac{\gamma_{\Delta}}{|S|}) \geq \text{log} \frac{f(S^{t - 1})}{f(\varnothing)} .
	\end{split}
	\end{align}
	Based on the fact that $ \text{log}(x) \geq  1 - 1/x$, $\forall x \geq 1$, we derive
	\begin{align}\notag
	\begin{split}
	\text{log} \frac{f(\varnothing)}{f(S^{t - 1})}  \geq   |S^*|\text{log}(1 - \frac{\gamma_{\Delta}}{|S|})^{-1} \geq |S^*| \frac{\gamma_{\Delta}}{|S|} \, \text{,}
	\end{split}
	\end{align}
	\setlength{\parindent}{0em}implying that
	\begin{align}\notag
	\begin{split}
	\gamma_{0} = \frac{1}{\gamma_{\Delta}} \text{log} \frac{f(\varnothing)}{f(S^{t - 1})}   \geq  \frac{|S^*|}{|S|} .
	\end{split}
	\end{align}
	Furthermore, due to $|S^*| \geq k_{\text{min}}$ and $|S| \leq k$, we have
	\begin{align}\label{bound}
	\begin{split}
   			e^{- \frac{k_{\text{min}}}{k} \gamma_{\Delta}} f(\varnothing) 	\geq	e^{-\frac{|S^*|}{|S|}\gamma_{\Delta}}f(\varnothing)  \geq  f(S^{t-1}).
	\end{split}
	\end{align}
	The proof is complete.
\end{proof}

\setlength{\parindent}{1em}It is noticeable that the optimality bound $\gamma_{0} \in (0\text{,}\,1)$ can be calculated after leader selection. Besides, this bound can be as small as possible under certain parameter matrices, which implies that the solution approximates the optimal one nearly. In \eqref{bound}, when $\gamma_{\Delta}$ is larger, which points out that the submodularity of $f$ appears significantly, the value of $f(S^{t-1})/f(\varnothing)$ is smaller. Thus, it is inferred that the increment of adding a leader is larger, which is beneficial to satisfy $f = 0$ with less leaders.

Actually, there exists conservativeness for Algorithm 1, which is analyzed in the next section. Here, we present another one to reduce the conservativeness.  

\begin{algorithm}[H]
	\caption{Algorithm for selection of a minimum-size \emph{leader set} $\tilde{S}$ with a set of given TDDT to assure the tracking of a reference signal}
	\begin{algorithmic}[1]
		\Require   The agents index set $ \Omega$, the metric $f$, a constant $k$, $T_{\text{min-max}}$ and $Q$
		\Ensure The \emph{leader set} $\tilde{S}$
		\Procedure{MinSet $(\tilde{S}, f)$}{}
		\State \textbf{Initialization:} \;\;\;\; $\tilde{S}  \leftarrow S_{0}  $ ,  $z  \leftarrow  0  $ 
		\quad  \quad \While{$|\tilde{S}| \le k$}
		\quad  \quad \While{$z < z_{\text{max}}$ }  \\
		\quad  \quad   \quad \quad \textbf{If} \eqref{Theorem_5}-\eqref{Theorem_8} hold with a set of $K_p$ \\
		\quad  \quad      \quad \quad \quad \textbf{return} $\tilde{S}$\\
		\quad  \quad    \quad \quad \textbf{else} \\
		\quad  \quad	\quad	\quad \quad 	$z = z + 1$ with a new set $Q$	
		\quad  \quad    \EndWhile \\
		\quad  \quad     \textbf{If} $f = 0$ \\
		\quad  \quad    \quad  \textbf{return} ``None with such a $T_{\text{max-min}}$"  \\
		\quad  \quad    \textbf{else} 
		\quad \quad \For{ $v_x$ $\in \Omega \backslash  \tilde{S} $}
		\quad  \quad      \State $F_{v_x}$  $\leftarrow$  $f(\tilde{S}) - f(\tilde{S}\cup \{v_x\})    $ 
		\quad  \quad    \quad \quad \EndFor \\
		\quad  \quad \quad \quad \quad $v^*$   $\leftarrow$  arg $\text{max}_{v_x}$  $F_{v_x}$\\
		\quad  \quad \quad \quad \quad  $\tilde{S}$   $\leftarrow$   $\tilde{S} \cup \{v^*\}$  \\
		\quad  \quad \quad \quad \quad $z = 0$
		\quad  \quad    \EndWhile 
		\EndProcedure
	\end{algorithmic}
	
\end{algorithm}

The complexity of Algorithm 2 is $O(\xi z_{\text{max}}n^3)$, where $\xi = |\tilde{S}| - |S_{0}| + 1 \ge 1$, while the complexity of Algorithm 1 is $O(z_{\text{max}}n^3)$, which shows Algorithm 1 requires less running time when $\xi > 1$.


\section{EXAMPLES}

In this section, numerical examples are provided to verify the effectiveness for Algorithm 1. Firstly, we give the description for cases, including parameters settings. Then, the results are shown, composing of the tracking error evolution curve of \emph{followers}, comparison with two selection methods and the relation between number of leaders needed and the dwell time. Finally, we describe the conservativeness analysis.
\subsection{Cases Statement}

We consider a high-order linear multi-agent network with six agents and three predefined topologies, shown in Fig. 1. The individual self-dynamical matrix  $A$ \cite{Ni-2013-automatica} is
$$
\begin{gathered}
A =
\begin{pmatrix} 0.4147 & -0.4087 & -0.1287 \\ 0.3802 & -0.3380 & -0.3305 \\ 0.1313 & -0.7076 &  0.0233  \end{pmatrix} \text{,}
\quad
\end{gathered}
 $$
where it is not Hurwitz stable, $\lambda_{1}(A) = -0.50$, $\lambda_{2}(A) = 0.30 + 0.10i  $, $\lambda_{3}(A) = 0.30 - 0.10i $. The initial state of every agent is generated randomly within the range $(-100, 100)$. Furthermore, the switching law $\sigma(t)$ is aperiod.

\begin{figure}[H]\label{topology}
	\centering
	\captionsetup{justification = raggedright}
	\captionsetup{font = {footnotesize}}
	\includegraphics[height=3.5cm,width=8cm]{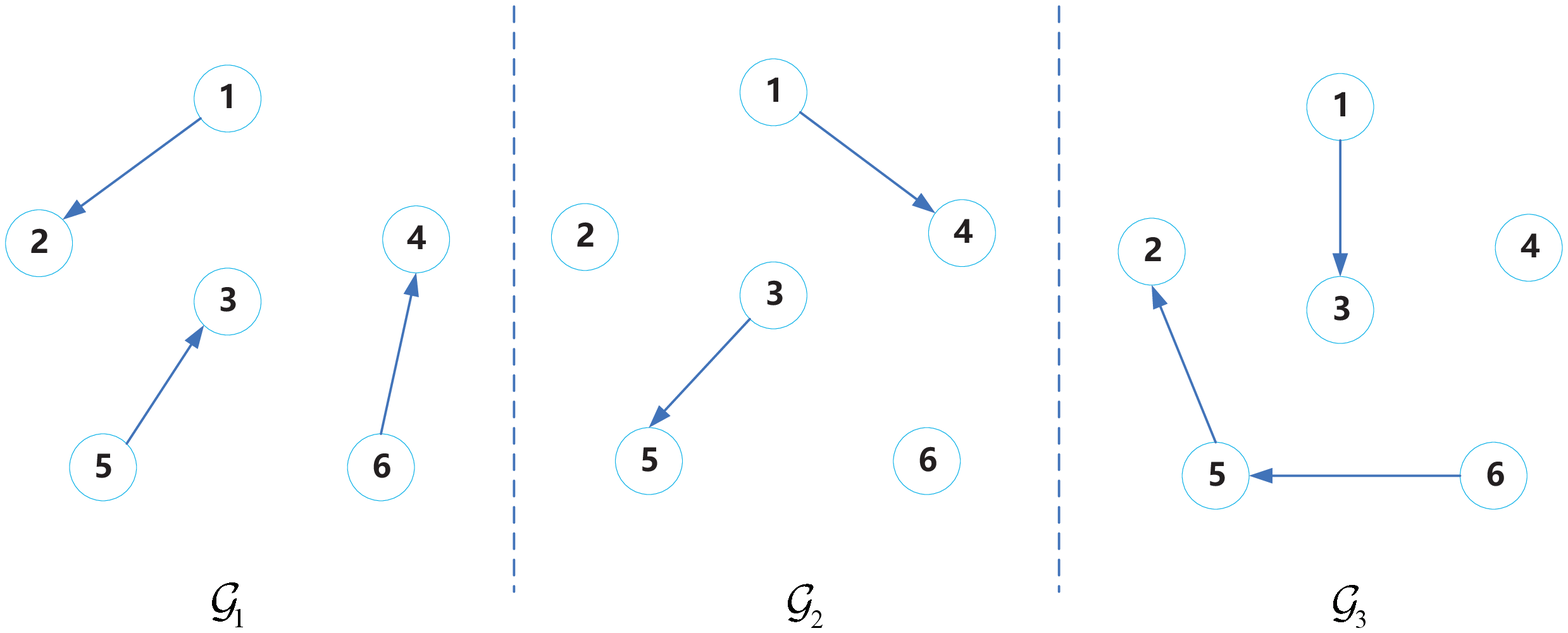}
	\caption{Three predefined interaction topologies. It is intuitive that there exist agents (agent 1 and agent 6) unreachable in the union of the directed interaction topologies before leader selection. After Algorithm 1, every mode satisfies $\lambda_{\text{r}}(\mathcal{A}_p) > 0$ based on the parameters in our simulation, where $S = \{1,5,6\}$. Specifically, $\lambda_{\text{r}}(\mathcal{A}_{\mathcal{G}_1}) \approx 0.05$, $\lambda_{\text{r}}(\mathcal{A}_{\mathcal{G}_2}) \approx 0.30$, $\lambda_{\text{r}}(\mathcal{A}_{\mathcal{G}_3}) \approx 0.30$, it signifies that each mode is unstable.}
\end{figure}
\subsection{Leader Selection Numerical Examples}

\setlength{\parindent}{1em} In this subsection, we depict the results of numerical examples for Algorithm 1. In the first place, we set the initial range of a set of given TDDT as $\tau_{\mathcal{G}_1} \in [1.00\text{,}\,2.00]$, $\tau_{\mathcal{G}_2} \in [0.50\text{,}\,1.50]$, $\tau_{\mathcal{G}_3} \in [0.50\text{,}\, 1.50]$. $k = 3$. In view of Algorithm 1, we acquire parameters as follows:

 \begin{equation}\notag
 \left\{
 \begin{array}{lr}
 l_{\mathcal{G}_1} = 3\text{,} \; \mu_{\mathcal{G}_1} = 0.03\text{,} \; \eta_{\mathcal{G}_1} = 2.0 \text{,} \; \\
 l_{\mathcal{G}_2} = 2\text{,} \; \mu_{\mathcal{G}_2} = 0.02\text{,} \; \eta_{\mathcal{G}_2} = 4.2 \text{,} \; \\
 l_{\mathcal{G}_3} = 2\text{,} \; \mu_{\mathcal{G}_3} = 0.04\text{,} \; \eta_{\mathcal{G}_3} = 2.8 \text{,} \;
 \end{array}
 \right.
 \end{equation}
 a \emph{leader set} $S$ = $\{1\text{,}\, 5\text{,}\, 6\}$, and the practical range of TDDT:
 \begin{equation}\notag
 \left\{
 \begin{array}{lr}
 \tau_{\mathcal{G}_1}^* \in [1.60\text{,}\,1.74]  \text{,}\\
 \tau_{\mathcal{G}_2}^* \in [0.83\text{,}\,0.92] \text{,} \\
 \tau_{\mathcal{G}_3}^* \in [0.94\text{,}\,1.13] \text{,}
 \end{array}
 \right.
 \end{equation}
 where the unit of time is \emph{second}. Due to \eqref{Theorem_2}, we acquire the input gain matrices: 
 
 $$
 \begin{gathered}
 K_{\mathcal{G}_1} =
 \begin{pmatrix} 0.2275 & -0.0017 & 0.0002 \\ -00017 & 0.1399 & -0.0604 \\ 0.0002 & -0.0604 &  0.1819  \end{pmatrix} \text{,}
 \quad
 \end{gathered}
 $$
 
$$
\begin{gathered}
K_{\mathcal{G}_2} =
\begin{pmatrix} 1.2657 & -0.0143 & 0.0013 \\ -0.0143 & 0.5130 & -5190 \\ 0.0013 & -0.5190 &  0.8743  \end{pmatrix} \text{,}
\quad
\end{gathered}
$$

$$
\begin{gathered}
K_{\mathcal{G}_3} =
\begin{pmatrix} 1.2763 & -0.0003 & 0.0000 \\ -0.0003 & 1.2612 & -0.0104 \\ 0.0000 & -0.0104 &  1.2685  \end{pmatrix} \text{.}
\quad
\end{gathered}
$$

\setlength{\parindent}{0em}In what follows, the tracking error evolution curves of \emph{followers} are shown in Fig. 2. Although we set no requirement on connectivity, based on Algorithm 1, after a number of trails, we find that each agent is reachable in one topology at least. This signifies that each agent is reachable in the union of the directed interaction graphs, which is the necessary topology condition for the realization of the tracking.  

 \begin{figure}[t]\label{evolution}
 	\centering
 	\captionsetup{justification = raggedright}
 	\captionsetup{font = {footnotesize}}
 	\includegraphics[height=13cm,width=8.5cm]{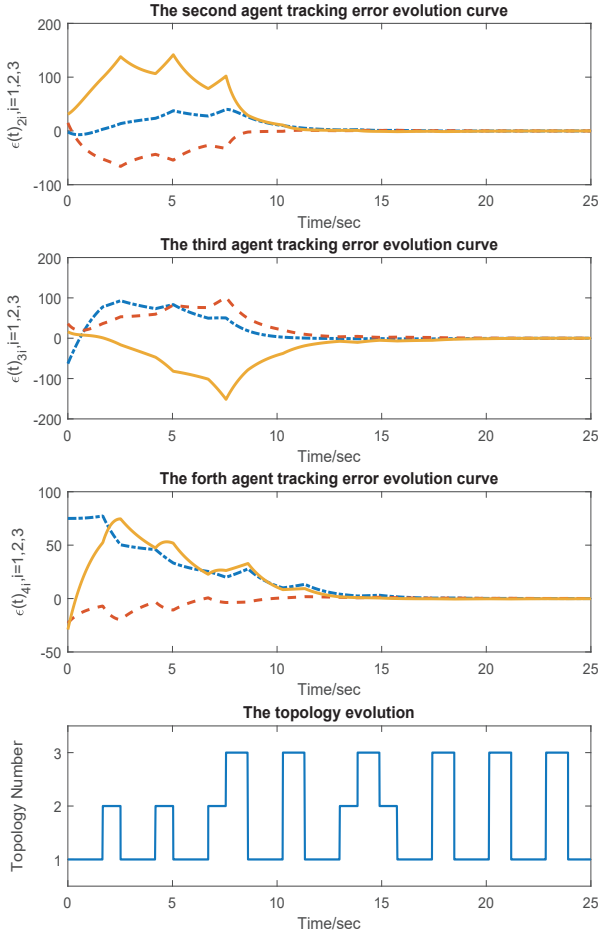}
 	\caption{Followers tracking error evolution. In view of the \emph{leader set} $S$ = $\{1, 5, 6\}$, the former three subgraphs depict the tracking error evolution process of agent $2$, $3$, and $4$ respectively. It is obvious to see that the switching law is aperiod. 
 	}
 \end{figure}

\setlength{\parindent}{1em} In the second place, we show the comparison with other two selection methods, where parameters are based on those that are mentioned above. We direct at greedy selection with the index of $f_{\text{max}} = \sum_{p \in \mathcal{P}}^{}$$\lambda_{\text{r}}(\mathcal{A}_p^{(2)})$, as well as random selection with $f$. To be precise, we construct the metric $f_{\text{max}}$ which is anticipated to be minimized, since it is an intuitive measure to fulfill the decisive factor \eqref{Theorem_2}. For greedy selection with $f_{\text{max}}$, we still draw on the proposed method, but the metric $f$ is replaced with $f_{\text{max}}$. For random selection, we refer to Algorithm 1, but we select a new leader from $\Omega \backslash S$ randomly rather than with the greedy rule. For the persuasiveness of this scheme, we take the expected value of 100 trails. The significance of this simulation between diverse three methods, lies in the comparison of optimization performance when we add an agent to $S$. Concretely, we explore $f_{\text{max}}$ from the perspective of the maximal real eigenvalue, while $\gamma$-submodular function $f$ from the Euclidean distance of controllability. The better performance implies $f = 0$ with less leaders. Then, the result is shown in Fig. 3. As is shown, the optimization performance of the proposed method is optimal comparatively.  

\begin{figure}[t]\label{comparison}
	\centering
	\captionsetup{justification = raggedright}
	\captionsetup{font = {footnotesize}}
	\includegraphics[height=4.2cm,width=7.0cm]{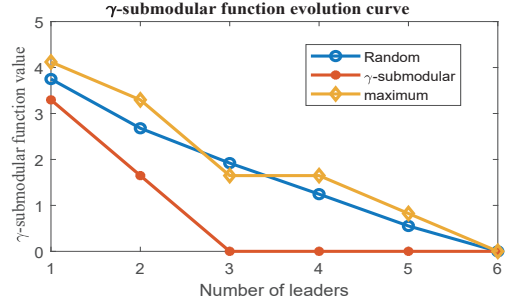}
	\caption{Comparison of three methods. \emph{Random} corresponds to random selection with the metric $f$, selecting a new leader randomly for each iteration. $\gamma$-\emph{submodular} represents our method. \emph{Maximum} denotes greedy selection, with the metric $f_{\text{max}}$, adding a leader to minimize the index. It is shown our algorithm is optimal comparatively, 
	Especially, it is of interest to mention that curves of $\gamma$-\emph{submodular} and \emph{maximum} almost coincide with each other, with a certain set of $K_p$, $p \in \mathcal{P}$. For instance, $K_{\mathcal{G}_1} = 0.45 \cdot I_{Nn}$, $K_{\mathcal{G}_2} = 3.15 \cdot I_{Nn}$, $K_{\mathcal{G}_2} = 2.75 \cdot I_{Nn}$. 
}
\end{figure}


\begin{figure}[t]\label{connection}
	\centering
	\captionsetup{justification = raggedright}
	\captionsetup{font = {footnotesize}}
	\includegraphics[height=4.2cm,width=7.6cm]{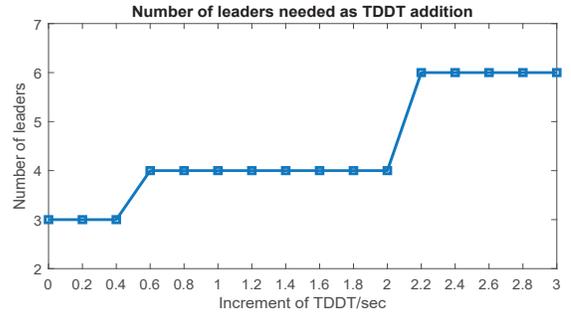}
	\caption{Relation between leader selection and TDDT. For brevity, we set the lateral axis as the increment of TDDT, which implies that the lateral axis value of each point represents the augmentation for the TDDT based on the parameters in the first numerical case.   }
\end{figure}

\renewcommand\arraystretch{2.5}
\begin{table*}[t] 
	\centering	
	\caption{COMPARISON AS NUMBER OF STABLE MODES ADDITION} 
	\begin{tabular}{p{5.17cm}<{\centering}||p{5.17cm}<{\centering}||p{5.17cm}<{\centering}} 
		\hline
		\hline
		$\lambda_{\text{r}}(\mathcal{A}_{\mathcal{G}_1}) < 0$, $\lambda_{\text{r}}(\mathcal{A}_{\mathcal{G}_2}) > 0$, 
		$\lambda_{\text{r}}(\mathcal{A}_{\mathcal{G}_3}) > 0$,& 	$\lambda_{\text{r}}(\mathcal{A}_{\mathcal{G}_1}) < 0$, $\lambda_{\text{r}}(\mathcal{A}_{\mathcal{G}_2}) < 0$, 
		$\lambda_{\text{r}}(\mathcal{A}_{\mathcal{G}_3}) > 0$, & 	$\lambda_{\text{r}}(\mathcal{A}_{\mathcal{G}_1}) < 0$, $\lambda_{\text{r}}(\mathcal{A}_{\mathcal{G}_2}) < 0$, 
		$\lambda_{\text{r}}(\mathcal{A}_{\mathcal{G}_3}) < 0$, \\ 
		\hline 
		$S = \{1\text{,}\; 5\text{,}\; 6\}$ & 	$S = \{1\text{,} \; 2 \text{,} \; 3 \text{,}\; 5\text{,} \; 6\}$  & 	$S = \{1\text{,} \; 2 \text{,} \; 3 \text{,}\; 4 \text{,}\; 5\text{,} \; 6\}$   \\
		\hline
		$\tau_{\mathcal{G}_1}' > 1.77\text{,} $ 	
		$\tau_{\mathcal{G}_2}' = 0.65\text{,} $
		$\tau_{\mathcal{G}_3}' = 0.85\text{,} $
		& 	
		$\tau_{\mathcal{G}_1}' > 1.56\text{,} $ 	
		$\tau_{\mathcal{G}_2}' > 0.54\text{,} $
		$\tau_{\mathcal{G}_3}' = 0.62\text{,} $ & 
		
		$\tau_{\mathcal{G}_1}' > 1.51\text{,} $ 	
		$\tau_{\mathcal{G}_2}' > 0.51\text{,} $
		$\tau_{\mathcal{G}_3}' > 0.51\text{,} $   \\
		\hline
		\hline
	\end{tabular}
\end{table*}

In the third place, we investigate the relation between the number of leaders needed and the TDDT. As the basis of the above-mentioned parameters configuration, we adjust the $\eta_p$ to alter the TDDT. Specifically, we operate Algorithm 1 at every turn adding 0.2 second. Thereby, we obtain the result shown in Fig. 4. This result shows the relationship between leader selection and the dwell time. 
It is straightforward to see that all the agents are required to be accessible to the reference signal, when the increment of the TDDT is over 2.2 second. In addition, through a number of trails, it is interesting to point out that there exists a proportion relation between the eligible TDDT. This means that each TDDT is not allowed to differ one other greatly. We consider that it is mainly due to the connectivity between predefined topologies as well as the feedback gains matrices.

In the last place, based on Algorithm 1, we investigate the comparison between different number of stable modes, including the result of \emph{leader set} and corresponding TDDT. The result is shown in the Table \uppercase\expandafter{\romannumeral1}. As is shown, when the number of stable modes increases, the number of leaders needed grows. However, when we consider each mode unstable, we require leaders with the minimum number to ensure the tracking.

\subsection{Conservativeness Analysis}

\setlength{\parindent}{1em} In this subsection, we illustrate the conservativeness analysis of Algorithm 1. In Fig. 1, there are two unreachable agents in the union of the directed interaction graphs, agent 1 and agent 6 respectively. We set $\{1\text{,}\, 6\}$ as the \emph{leader set}, by virtue of Lemma \ref{based-Xiang}, and then we can acquire a set of feasible TDDT to ensure the tracking. However, if we take the same parameters to execute Algorithm 1, we can not acquire the \emph{leader set} $\{1\text{,}\, 6\}$, but $\{1\text{,}\, 5\text{,}\,6\}$. Thus, we consider the conservativeness is due to without considering the impact of union of the directed interaction topologies in Theorem \ref{unstable-Xiang} and Lemma \ref{lemma-K}. Then, it causes that each part $f_p$ of the leader selection metric $f$ has to be satisfied separately, where
\begin{align}\notag
\begin{split}
f = \sum_{p\in \mathcal{P}}  f_p \text{,} \;f_p  = \sum_{i:\text{Re}_{f}}  \text{dist}^2(v_i, \text{span}(W(\hat{\mathcal{A}}_p, \hat{B}))) \text{.}
\end{split}
\end{align}

For instance, for the first consequence in  \emph{Leader Selection Simulation}:
	\begin{align}\notag
		\begin{split}
			\mathcal{G}_{1}: \lambda_{\text{r}} (\hat{\mathcal{A}}_{\mathcal{G}_1}  )  & \approx \; \; \,0.2414 \text{,} \\
			\mathcal{G}_{2}:\lambda_{\text{r}} (\hat{\mathcal{A}}_{\mathcal{G}_2}  ) & \approx -0.5826  \text{,} \\
			\mathcal{G}_{3}:\lambda_{\text{r}} (\hat{\mathcal{A}}_{\mathcal{G}_3}   ) & \approx -0.0357    .
		\end{split}
	\end{align}
Prior to leader selection, $S = \varnothing$, then $f_{\mathcal{G}_1} = 3.2936$, $f_{\mathcal{G}_2} = 0$, $f_{\mathcal{G}_3} = 0$. In order to decease $f = f_{\mathcal{G}_1}$ to zero, we operate Algorithm 1 to obtain the result as $\{1\text{,}\, 5\text{,}\,6\}$. Obviously, this \emph{leader set} satisfies the condition that with least leaders, each agent is reachable in $\mathcal{G}_1$ instead of the union of the directed interaction graphs. In addition, the conservativeness analysis does not mean that it is enough to select the unreachable agents as leaders in the union of the directed interaction topologies. Actually, the system \eqref{error-system} requires more leaders to guarantee the tracking when the TDDT increases, such as the result in Fig. 4. To reduce the conservativeness, we propose an heuristic method as Algorithm 2. By such a scheme, with parameters in the \emph{Case Statement}, we obtain the result $\tilde{S} = \{1\text{,}\,6\}$, as well as the corresponding TDDT:$\tau_{\mathcal{G}_1} = 1.54$, $\tau_{\mathcal{G}_1} = 0.76$, $\tau_{\mathcal{G}_3} = 1.18$. In this example, $\xi = |\tilde{S}| - |S_{0}| + 1 = 1$. Besides, because of the conservativeness in the condition $f = 0$, when this requirement is removed, then it is definite that the $|\tilde{S}|$ returned by Algorithm 2 is less than or equal to $|S|$ returned by Algorithm 1. 


\section{Conclusion}

In this paper, we investigate the problem of choosing a minimum-size \emph{leader set} to achieve the tracking of a reference signal with a set of given TDDT. We show the problem description as $\bar{\mathcal{P}}1$. Then, we present Theorem \ref{unstable-Xiang} to assure the desired tracking. Subsequently, we establish a metric based on the proposed sufficient condition, and then we formulate the combinatorial optimization problem as $\check{\mathcal{P}} 2$. We design Algorithm 1 with the provable optimality bound to deal with $\check{\mathcal{P}} 2$. Then, we propose Algorithm 2 to reduce the conservativeness but the complexity may be higher than Algorithm 1. Finally, we show the numerical cases to evaluate effectiveness of the proposed method, and the conservativeness analysis for the algorithm is provided. 
The switching \emph{leader set} is considered as our future work, while the \emph{leader set} is fixed in this paper.

\appendices
\section{}


{\color{black}\begin{proof}[Proof of Lemma \ref{lemma-eigenvalue}]
		It is intuitive to see
		\begin{align}\notag
			\text{Re}(\lambda_{i}( A_p^{(1)} - \alpha I_{Nn})) < 0 \text{,} \; \text{where}\; \alpha =  - \frac{l_p + \varphi}{2 \beta \tau_p^{\text{min}}} \text{,} \; \forall p \in \mathcal{P} .
		\end{align}		
	Thus, we derive that
		\begin{align}\notag
		\text{Re}(v_i^T (A_p^{(1)} - \alpha I_{Nn} ) v_i) < 0 \text{,}
		\end{align}
		where $v_i$ is the $i$th eigenvector of $(A_p^{(1)} - \alpha I_{Nn})$ corresponding to $\lambda_i((A_p^{(1)} - \alpha I_{Nn}))$, and $v_i^Tv_i = 1$. In addition, it is straightforward to see that $v_i$ is also the $i$th eigenvector of $A_p^{(1)}$ corresponding to $\lambda_i(A_p^{(1)})$, and then $v_i^T A_p^{(1)} v_i = \lambda_{i}(A_p^{(1)})$. Then, based on the analysis above, we obtain
		\begin{align}\notag
		& \, \text{Re}(\lambda_i(A_p^{(1)} - \alpha I_{Nn})) \notag \\
	    = & \, 	\text{Re}(v_i^T (A_p^{(1)} - \alpha I_{Nn}) v_i) \notag \\
	    = &	\, \text{Re}(v_i^T A_p^{(1)} v_i) - 	\text{Re}(v_i^T \alpha I_{Nn} v_i) \notag \\
	    = & \,  \text{Re}(v_i^T A_p^{(1)} v_i) - \alpha \notag  \\
	    = & \, \text{Re}(\lambda_i(A_p^{(1)}))  - \alpha \notag \\
		< & \, 0 \notag  \text{,}
		\end{align}
Thereby, we have 
	\begin{align}\notag
		\text{Re}(\lambda_{\text{r}}(A_p^{(1)})) < \alpha . 
	\end{align}
The proof is complete.
\end{proof}}


%
%

\ifCLASSOPTIONcaptionsoff
  \newpage
\fi

\end{document}